\documentclass[12pt,a4paper]{amsart}

\usepackage{amsmath}
\usepackage{amsthm}
\usepackage{amssymb}
\usepackage{amsfonts}
\usepackage[matrix,arrow,curve]{xy}

\newtheorem{theorem}{Theorem}[section]
\newtheorem{lemma}[theorem]{Lemma}
\newtheorem{proposition}[theorem]{Proposition}

\theoremstyle{definition}

\usepackage{latexsym}
\usepackage{amsmath}
\usepackage{amsfonts}
\usepackage{amssymb}

\newcommand{\be}{\begin{equation}}
\newcommand{\ee}{\end{equation}}
\newcommand{\bea}{\begin{eqnarray}}
\newcommand{\eea}{\end{eqnarray}}

\newcommand{\ba}{\begin{aligned}}
\newcommand{\ea}{\end{aligned}}

\def\cS{{\mathcal S}}                       %
\def\RR{{\mathbb R}}                    %
\def\CC{{\mathbb C}}                    %
\def\XX{{\mathbb X}}                    %
\def\1{{\mbox{\boldmath $1$}}}          %
\def\tr{\mathrm{tr\,}}                  %
\def\im{\mathrm{Im\,}}
\def\0{{\mbox{\boldmath $0$}}}          %
\def\TT{{\mathbb T}}                    %
\def\g{{\mathfrak g}}  \def\b{{\mathfrak b}} \def\k{{\mathfrak k}}

\def\half{\textstyle{\frac{1}{2}}}

\def\t{{\mathbf{t}}}  \def\tilt{\tilde{\mathbf{t}}} \def\tv{\tilde v}
\def\dt{\left.\frac{d}{dt}\right|_{t=0}}

\def\refcom{}
\def\dagt{{T}}
\def\sgn{\rm{sgn}}

\title[A new model in the Calogero-Ruijsenaars family]{A new model in the Calogero-Ruijsenaars family}

\begin{document}

\maketitle
\author{Ian Marshall}

\address{Faculty of Mathematics, Higher School of Economics, Ulitsa Vavilova 7, Moscow}

\begin{abstract}
Hamiltonian reduction is used to project a trivially integrable system on the Heisenberg double of $SU(n,n)$, to obtain a system of Ruijsenaars type on a suitable quotient space. This system possesses $BC_n$ symmetry and is shown to be equivalent to the standard three-parameter $BC_n$ hyperbolic Sutherland model in the cotangent bundle limit.
\end{abstract}

\section{Introduction}

Consider the Hamiltonian of Ruijsenaars type,
\be\label{openingformula}
\begin{aligned}
&H(q,p) = a^2\sum_{i=1}^ne^{-2q_i}\\
&-
\sum_{i=1}^n (\cos p_i)
\left[{ 1 + \left({1+ b^2}\right)e^{-2q_i} + b^2e^{-4q_i}}\right]^{\half}
\prod_{k\neq i}\left[ 1 - \frac{c^2}{4\sinh^2(q_i-q_k)}\right]^{1/2},
\end{aligned}
\ee
for arbitrary positive constants $a^2, b^2$, and $c^2$. 
The $q_k$ are restricted by a condition of the form
$q_1 > q_2 +s,\ q_2 > q_3 +s,\, ...$  with a suitable $s > 0$, so as to guarantee that
the arguments of all the square-roots in the product are positive.

The integrability of this Hamiltonian will be shown via a Poisson Lie group analog of the reduction argument of Kazhdan-Kostant-Sternberg.

\subsection{Some history}

Following the initial works of Calogero, Moser and Sutherland, Olshanetsky and Perelomov undertook a study of the family of Calogero-type systems, these being characterised by different classical root systems, and different choices for the interaction function. See \cite{Calogero, Moser, Sutherland, OP1, OP2}.
\refcom
Being characterised by the different root systems means that each system is invariant with respect to the action of the appropriate Weyl group. The most general interaction function is the Weierstrass $\wp$\,-\,function, and the most general root system is the $BC_n$ one: all other integrable cases come from these via different limits and special choices of parameters. Actually this is not quite true, as the $A_n$ case sits apart, but it holds true for $(B_n, C_n, D_n)$ root systems, and was a strong indication of a universal property for the $BC_n$ model. However, in order to include all root systems as special cases of the $BC_n$ one it is necessary that there be three independent parameters, whilst Olshanetsky and Perelomov were unable to prove integrability unless only two of these are independent. Following the Olshanetsky and Perelomov work, there was therefore an important missing gap. This deficiency was later repaired \refcom  by Inozemtsev and Meshcheryakov \cite{Inozemtsev}, who showed that the $BC_n$ case with three independent parameters was indeed integrable, which justifies its role as a universal model.

A crucial tool for the Olshanetsky and Perelomov approach was symmetry reduction. They looked at several examples of Lie groups and at cases of trivially integrable flows on them, admitting large groups of symmetries. This permits projection to the corresponding homogeneous space and the result is sure to still be integrable. All of their examples arose in this way. The more general three-parameter result 
was not obtained by symmetry reduction, and so another missing gap remained to be filled: How might one obtain the general integrable $BC_n$ case by a reduction procedure?

Shortly after the work of Olshanetsky and Perelomov appeared, Kazhdan, Kostant and Sternberg produced an article \cite{KKS} which has become one of the classic references for this subject. \refcom They used the method of canonical Hamiltonian symmetry reduction to derive the rational and trigonometric $A_n$ models. Their result remains one of the most famous examples of the use of canonical symmetry reduction, and now appears routinely in textbooks. This method of reduction is usually called ``Marsden-Weinstein reduction'', whilst aficionados will often smile wisely and say ``Of course, this was all known to Lie...'' and it surely has other representations which might vie for priority. The main thing is that as we inch our way towards a better understanding of models like the one discussed here, such methods are important and we can all enjoy the pleasure of using them, while acknowledging that their formulation is owed to several different authors. Indeed the KKS result itself emerged at a time when the reduction approach was asserting itself as especially effective, and it marks a kind of landmark in the general subject of reduction just for that reason. A key role in the KKS result was that of the fixed image taken by the momentum map: this is now often referred to as ``the KKS element'' and  it will appear, in no less crucial a role, in the present article as well.

Since the time of these early articles the collection of systems going variously under the name of Calogero-, Caloger-Moser-, Sutherland- and so on, have been studied for a host of different reasons, and have proved to be a rich source of mathematical treasures.

In 2005 \cite{FehPus} Feher and Pusztai filled a missing part of the story, by showing how to obtain the general hyperbolic $BC_n$ case by canonical symmetry reduction. \refcom Their result is important in relation to the present article, as it uses the setting of the cotangent bundle of $SU(n,n)$ and this is the picture generalised here. 

Meanwhile, in the 1990s Ruijesaars and Schneider had initiated a study of relativistic versions of the models in the Calogero family \cite{RS}. \refcom These were all models of $A_n$ type and they are invariant with respect to the Poincar\'e group, rather than the Galilei group, as had been the case for the standard systems previously studied. A number of interesting articles by Ruijsenaars appeared on this theme, see especially \cite{SR-CRM}, some hinting at a possible reduction approach to their description. In due course, in an article by Gorsky and Nekrasov \cite{GN}, a reduction construction was given to describe one of the Ruijsenaars models, with various suggestions for relevant modifications of the standard point of view used and of the KKS result, and these suggestions were elaborated in several subsequent articles, for example \cite{G, AFM}. \refcom Shortly afterwards Fock and Rosly  showed that the complex trigonometric Ruijsenaars model arises via a canonical reduction construction using a suitably defined classical r-matrix and Poisson Lie groups \cite{FockRosly}, see also \cite{FGNR}. \refcom What is especially interesting in the Fock and Rosly paper is their treatment of the so-called ``dual system'' which appeared in the articles of Ruijsenaars and which arises in a natural way when viewed in the context of reduction.
The reduction viewpoint has been developed in several articles by Feher and Klimcik \cite{FK}, in which they have obtained various Ruijsenaars models as reductions based on the use of Poisson Lie groups, and also derived global results concerning dual systems.

Several non-$A_n$ Ruijsenaars type deformations of Calogero-Moser-Sutherland models were found by van Diejen in \cite{vD}. It is interesting that the \emph{rational} $BC_n$ model of van Diejen is dual to the \emph{hyperbolic} $BC_n$ Sutherland model. This was proved in \cite{Pus} by Pusztai, by the use of reduction.

\subsection{Brief summary of the result}
The present article is a contribution to the project of describing Calogero-type systems, or more particularly of the Ruijsenaars ones. It describes the Poisson Lie group version of the Feher-Pusztai result, and gives rise to what is claimed to be a \emph{Ruijsenaars type deformation of the Hyperbolic $BC_n$ Sutherland model}. The justification for this is that it has all the standard properties of a Ruijsenaars system, and yields the classical three-parameter hyperbolic $BC_n$ system in the non-relativistic limit.

The standard Ruijsenaars type deformation of the hyperbolic $BC_n$ Sutherland model is the one found by van Diejen \cite{vD}. \refcom The system obtained in the present article appears to be different to that one. The reasons for this are not clear. The origins of the hyperbolic van Diejen model via reduction are unknown, so now there is some new mystery to uncover. 

Feher and Pusztai begin by considering the group $K=SU(n,n)$ with maximal compact subgroup denoted by $K_+$. The momentum map for the symplectic lift of the natural right action of $K_+$ on $K$ to the cotangent bundle $T^*K$ is fixed to be a character of $K_+$. The ensuing reduction yields the cotangent bundle of the symmetric space $K/K_+$
(with the canonical symplectic structure modified by the addition of a ``magnetic term'').
This is the secret of the Feher-Pusztai result, as the the set of characters of $K_+$ is a one dimensional space, and the parameter on this space becomes the extra parameter in the result. The result of Olshnetsky and Perelomov, although it was cast in a different language, had effectively ignored the possibility that $T^*(K/K_+)$ could have emerged at anything other than the zero value of that parameter. The rest of the reduction involves the momentum map for the symplectic lift of the natural action of $K_+$ on $K/K_+$ and this is fixed to be a suitable analog of the KKS element. What follows is a matter of direct computations: a concrete formula can be computed for the symplectic structure on the reduced space, and a family of trivially commuting Hamiltonians on $T^*K$, all of which are invariant with respect to both actions of $K_+$, descend to give a commuting family on the reduced - or projected - space. The simplest Hamiltonian in this family is the $BC_n$ Hamiltonian of interest.

The result to be described in the present article follows the Feher-Pusztai one rather closely. First of all the cotangent bundle $T^*K$ is replaced by the Heisenberg double $G$. Then the momentum maps for the left and right actions of $K_+$ on $G$ are fixed, both in ways which precisely imitate the values they were fixed to be in the cotangent case. The reduction is made, again amounting in practice just to a concrete computation, to obtain a formula for the symplectic structure on the reduced space, and then a family of trivially commuting Hamiltonians on $G$ is seen to descend to the reduced space, where in the concrete coordinates involved in the formula obtained for the symplectic structure, the first, simplest, of the commuting Hamiltonians is seen to be a Ruijsenaars type analog of the Feher-Pusztai one.

The model described in the present work \emph{is} of Ruijsenaars form, it \emph{is} invariant under the Weyl group for the $BC_n$ root system, it \emph{is} a deformation of the standard hyperbolic Sutherland model, and it deserves for all these reasons to be known. Here it is put forward as being moreover an interesting application of Poisson Lie group reduction.

\bigskip\noindent\textbf{Acknowledgement.}
The contribution of Laszlo Feher is gratefully acknowledged. This work represents a regrettably aborted collaboration, begun several years ago. The idea of generalising the KKS approach from cotangent bundle to Heisenberg double was explained to me by Feher, and I am grateful to him for his generous and expert advice. 

\section{The Heisenberg double of $SU(n,n)$}\label{firstmainsec}

Canonical symmetry reduction will be applied to the Heisenberg double of the Poisson Lie group $SU(n,n)$, so it is worthwhile to start off by describing what this space looks like. 
At the same time, some notation is established for use in the rest of the paper:

Let $I_{nn}$ be the $2n\times2n$ matrix $\left(\begin{matrix} \mathbf{I_n}&\0\\ \mathbf{0}&-\mathbf{I_n}\end{matrix}\right)$.
 
$G$ will denote $SL(2n,\CC)$.

$K$ denotes $SU(n,n)=\{g\in SL(2n,\CC)\ \vert\  g^\dagger I_{nn}g=I_{nn}\}$.

$K_+:=\left\{\left(\begin{matrix} p&\0\\ \0&q \end{matrix}\right)\right\} = K\cap SU(2n)$.

\medskip
$B$ denotes the set of all upper triangular matrices in $SL(2n,\CC)$ with real, positive
diagonal entries, and $B_n$ denotes the same set in $GL(n,\CC)$.

\medskip

$\TT$ denotes the diagonal subgroup in $U(n)$.\\

As a vector space, the Lie algebra $\g=Lie(G)$ can be decomposed as the sum $\g=\k+\b$ of the two subalgebras $\k=Lie(K)$ and $\b=Lie(B)$, with respect to which the projections $P_\k:\g\rightarrow\k$ and $P_\b:\g\rightarrow\b$ are well-defined. Let $\langle\ ,\ \rangle:\g\times\g\rightarrow\RR$ denote the non-degenerate, invariant inner product defined by
\[
\langle X, Y\rangle = \im\tr XY.
\]
Then $R:=P_\k-P_\b$ defines a classical r-matrix on $\g$, skew-symmetric with respect to 
$\langle\ ,\ \rangle$. For any function $F\in C^\infty(G)$, the left- and right-derivatives, $D^{l,r}F:G\rightarrow\g \sim\g^*$, of $F$ are defined by
\[
\dt F(e^{tX}ge^{tY}) = \langle D^lF(g),X\rangle + \langle D^rF(g), Y\rangle\qquad \forall X,Y\in\g.
\]
The Poisson structure on $G$, as the Heisenberg double based on the bi-algebra $\g=\k+\b$, is defined by 
\be\label{toppb}
\{F,H\} = \langle D^lF, R(D^lH)\rangle + \langle D^rF, R(D^rH)\rangle.
\ee
Let $c_1, c_2\in G$ and define the subspace $M(c_1,c_2)\subset G$ by
\[
M(c_1,c_2)= \{b c_1 k \ |\  b\in B, k\in K\} \cap \{k c_2 b\ |\ b\in B, k\in K\}.
\]
Introduce ``coordinates'' $(b_L,k_L,b_R,k_R)$ on $M(c_1,c_2)$ (which of course are not independent) by
\[
M(c_1,c_2)\owns g = b_Lc_1k_R = k_Lc_2b_R.
\]

\begin{proposition}\emph{(Alekseev and Malkin \cite{AleksMalk})}
$M(c_1,c_2)$ is a symplectic leaf, and all symplectic leaves are of this form.
The symplectic structure on $M(c_1,c_2)$ can be written in the form
\be\label{heisensymp}
[Symp](g) = \langle db_Lb_L^{-1}\,\overset{\wedge},\, dk_Lk_L^{-1}\rangle
+ \langle b_R^{-1}db_R\,\overset{\wedge},\, k_R^{-1}dk_R\rangle.
\ee
\end{proposition}

The inner product $\langle\ ,\ \rangle$ has been used so far to identify $\g^*$ with $\g$, and we may also use it to identify  $\k^*\sim\b^\perp\subset\g$ and $\b^*\sim\k^\perp\subset\g$; but for the same reason that $R$ is skew with respect to $\langle\ ,\ \rangle$, $\b^\perp=\b$ and $\k^\perp=\k$. The natural Poisson structures on $K$ and $B$ compatible with the Heisenberg double structure on $G$, defined above, are given by
\be\label{kpb}
\{\varphi,\psi\}(k) = \langle D^l\varphi(k)\, ,\, kD^r\psi(k)k^{-1}\rangle\qquad \phi,\psi\in C^\infty(K),\quad
D^{l,r}\varphi,\, D^{l,r}\psi \in \b\sim\k^*
\ee
and
\be\label{bpb}
\{\hat\varphi,\hat\psi\}(b) = \langle D^l\hat\varphi(b)\, ,\, bD^r\hat\psi(b)b^{-1}\rangle\qquad\hat\varphi,\hat\psi\in C^\infty(B),
\quad
D^{l,r}\hat\varphi,\, D^{l,r}\hat\psi \in \k\sim\b^*.
\ee
The meaning of (\ref{kpb}) and (\ref{bpb}) should be clear, but they are not necessary for the rest of the article, and are only included here for completeness. For a full account of the technology of Poisson Lie groups, the review by Reyman and Semenov-Tian-Shansky \cite{r_sts} \refcom is recommended: if the reader prefers one of the host of alternative texts to that one, it is readily and favourably endorsed here. Suffice it to say that $K$ and $B$ are Poisson Lie groups which act in a natural way on the \emph{Poisson space} $G$ (whose \emph{Lie group} structure is suppressed) and that one can reduce with respect to these actions by applying symplectic reduction (to the symplectic leaves).

\begin{proposition}\label{commies}
The functions $\Phi_\nu\in C^\infty(G)$ defined by
\[
\Phi_\nu(g) = - \frac{1}{2\nu}\tr\bigl(gI_{n,n}g^\dagger I_{n,n}\bigr)^\nu \qquad \nu=1,2,\dots
\]
are in involution with one another, with respect to the Poisson bracket (\ref{toppb}).
\end{proposition}

\begin{proof}
For all $\nu$, $D^l\Phi_\nu$ and $D^r\Phi_\nu$ are both in $\k$, and the result is a direct consequence of this fact.
\end{proof}

\begin{proposition} 
On $G$, the Hamiltonian vector field for $\Phi_1$ is integrated explicitly to give
\be\label{solhamvec1}
g(t) = g(0)\exp\bigl[-2itI_{n,n}g(0)^\dagger I_{n,n}g(0)\bigr]
\ee
\end{proposition}
\begin{proof} 
The left- and right-derivatives of $\Phi_1$ are computed to be 
$D^l\Phi_1(g) = -igI_{n,n}g^\dagger I_{n,n}$ and $D^r\Phi_1(g)=- iI_{n,n}g^\dagger I_{n,n}g$, with the result that the Hamiltonian vector-field can be written as
\be\label{hamvect1}
\dot g = -2g[iI_{n,n}g^\dagger I_{n,n}g] = -2[igI_{n,n}g^\dagger I_{n,n}] g = -2igI_{n,n}g^\dagger I_{n,n}g,
\ee
from which it follows that
\(
[igI_{n,n}g^\dagger I_{n,n}]
\)
is constant, and hence the result.
\end{proof}

Moreover, it follows from the fact that (\ref{hamvect1}) is a Hamiltonian vector-field, that it is tangent to symplectic leaves, and therefore that its integration to the curve (\ref{solhamvec1}), defined in the Proposition, lies entirely within one symplectic leaf. That is, $g(0)\in M(c_1,c_2)\Rightarrow g(t)\in M(c_1,c_2)\ \forall t$.

\section{The reduced space}
 In this section we shall explain the implementation of a reduction argument analogous to that of Kazhdan-Kostant-Sternberg (to which the abbreviation ``KKS'' will refer), adapted to the Heisenberg double. Specifically, what follows is the PLG version of a result of Feher and Pusztai, see \cite{FehPus}. 
 \refcom
 
Any element of $K$ may be written in the form
\be\label{diagk}
\left(\begin{matrix}\rho&\0\\ \0&\tau \end{matrix}\right)
\left(\begin{matrix} \Gamma&\Sigma\\ \Sigma &\Gamma \end{matrix}\right)
\left(\begin{matrix} k&\0\\ \0& l \end{matrix}\right) ,
\ee
with $\rho,\tau,k,l\in U(n)$, with $\Gamma=\cosh\Delta, \Sigma=\sinh\Delta$, and with $\Delta$ diagonal and real. Let us define the open subset $\check{K}$ of $K$ to be the set of all ``regular'' elements of $K$, being those which can be written in the above form, with the matrix $\begin{pmatrix}\0&\Delta\\ \Delta&\0\end{pmatrix}$ in the interior of a particular Weyl chamber within the Cartan subalgebra. That is, 
\be\label{Weylchambercondition}
\Delta= diag(\Delta_1,\dots,\Delta_n)\quad \hbox{with} \quad
\Delta_1 > \Delta_2 > \dots > \Delta_n.
\ee 
\textbf{Remark.} $\rho,\tau, k, l$ are only defined up to the action of $\TT$, given by $\delta\cdot(\rho,\tau, k, l) = (\rho\delta,\tau\delta, \delta^\dagger k, \delta^\dagger l)$. Restriction to $\check{K}$ is acceptable for the reason that we shall be treating such matrices as defining points in a Poisson space, on which a group action is defined, and with respect to which a reduction procedure may be applied: we shall \emph{not} be using all of their properties as members of a group, as we do not need to apply the group operation between regular elements.

\medskip
\noindent\textbf{Remark:} It is standard in discussions of Calogero type systems, to restrict to an open region of some fixed Weyl chamber, to avoid ``collisions of particles'', and this is reasonable as the forbidden points of the space correspond to poles of the potential.

\medskip

We can write any element of $B$ in the block-form,
$$
\left(\begin{matrix}\sigma&\nu\\ \mathbf{0}&\pi \end{matrix}\right),\qquad \sigma, \pi\in B_n,\ \nu\in gl(n,\CC).
$$

We restrict to the symplectic leaf $M(\mathbf{I},\mathbf{I})$; that is, to elements of $G=SL(2n,\CC)$ which may be written in the form
\be\label{leaf}
g=k_Lb_R=b_Lk_R\qquad\hbox{ with }\ k_{L,R}\in K,\ b_{L,R}\in B.
\ee
To be precise, we also restrict $k_L$ to be in $\check K$.

\subsection{The Constraints} By fixing $\sigma\in B_n$ and $x,y\in\RR_+$, constraints are imposed as follows: suppose that when written in the form $G\owns g=k_Lb_R$,
$b_R=\left(\begin{matrix}x\mathbf{I}&\omega\\ \0&x^{-1}\mathbf{I}  \end{matrix}\right)$ 
and that, when written in the form $g=b_Lk_R$, 
$b_L=\left(\begin{matrix} y^{-1}\sigma&y^{-1}\nu\\ \0&y\mathbf{I} \end{matrix}\right)$, with $\det(\sigma)=1$, and with both $\omega$ and $\nu$ undetermined in $gl(n)$. 
To start with $\sigma$ is just some fixed element in $B_n$, but it will shortly be specified to be the appropriate Poisson Lie Group analogue of the KKS element. These constraints are chosen in such a way that we may factor on the right of $G$ by $K_+$ and on the left by a big strict subgroup of $K_+$. However, as the map $g\mapsto b_R$ generates the right-action of $K$ on $G$ and the map $g\mapsto b_L$ generates the left-action of $K$ on $G$, both actions being Poisson, the fixing of the block-diagonal parts of $b_L$ and $b_R$, followed by projections to equivalence classes defined by the residual actions of the  isotropy subgroups of $K$, is sure to result in a reduced Poisson structure on the quotient space. 
The key property which ensures that this will work is that the left- and right-actions of $K_+$ on $G$ are \emph{admissible}. This property is explained in \cite{STSrims}, where it is shown to depend on a remarkably simple condition on the symmetry group: namely that the action of a subgroup of a Poisson Lie group is admissible if and only if the annihilator of its Lie algebra is a subalgebra in the dual Lie algebra. In the present context, we easily check that this condition is satisfied; indeed it says that 
$
\left[\begin{pmatrix}{\bf0}&*\\ {\bf0}&{\bf0}\end{pmatrix} , \begin{pmatrix}{\bf0}&*\\ {\bf0}&{\bf0}\end{pmatrix}\right] = \begin{pmatrix}{\bf0}&*\\ {\bf0}&{\bf0}\end{pmatrix}
$.
(I call this condition remarkable, for the reason that it does not depend on the space on which the group acts.)

As the family of functions defined in Proposition \ref{commies} are all invariant with respect to the left- and right-actions of $K_+$, they descend to functions on the reduced space, where their commuting property is preserved. The remaining sections amount to making a choice of local coordinates on the reduced space and seeing what the representatives of the commuting Hamiltonians $\Phi_k$ look like in these local coordinates.

\section{Computation of the symplectic structure on the reduced space}\label{reduction}

We'd like to obtain local coordinates on the reduced space and an expression for the reduced symplectic structure in terms of those coordinates. It follows from $g=b_Lk_R$, that
\be\label{ggdag}
gI_{nn}g^\dagger
=
b_LI_{nn} b_L^\dagger
=
\left(\begin{matrix}  y^{-1}\sigma&-y^{-1}\nu\\ \0&-y\mathbf{I}\end{matrix}\right)
\left(\begin{matrix} y^{-1}\sigma^\dagger&\0\\ y^{-1}\nu^\dagger&y\mathbf{I} \end{matrix}\right)
=
\left(\begin{matrix} y^{-2}(\sigma\sigma^\dagger - \nu\nu^\dagger) & -\nu\\ -\nu^\dagger&-y^2\mathbf{I} \end{matrix}\right)
\ee
After factoring on the right by $K_+$ and on the left by the subgroup $\left\{\left(\begin{matrix} \mathbf{I}&\0\\ \0&p \end{matrix}\right)\right\}\subset K_+$ it may be assumed that $g$ is in one of the two gauges:
$$
g=\left(\begin{matrix} \rho\Gamma&\rho\Sigma\\ \Sigma&\Gamma\end{matrix}\right)
\left(\begin{matrix} x\mathbf{I}&\omega\\ \0&x^{-1}\mathbf{I} \end{matrix}\right)
\quad\hbox{ with $\omega\in gl(n,\CC)$, $\rho\in SU(n)$}
$$
or
$$
g=\left(\begin{matrix} \rho\Gamma&\rho\Sigma\\ \Sigma&\Gamma\end{matrix}\right)
\left(\begin{matrix} k&\0\\ \0&l\end{matrix}\right)
\left(\begin{matrix} x\mathbf{I}&\omega\\ \0&x^{-1}\mathbf{I} \end{matrix}\right)
\quad\hbox{ with $\omega$ diagonal, real, positive and $\rho,k,l\in U(n)$.}
$$
In these two formulae $\Gamma$ and $\Sigma$ are as they were in (\ref{diagk}).

\noindent\textbf{Remark:} We have factored on the left by a subgroup of $K_+$. Later we shall factor by the remaining gauge freedom in $K_+$. More will be said about this later. Moreover, although we think of $\rho$ as living in $U(n)$, it is only defined up to right-multiplication by matrices in $\TT$ and this freedom will be exploited later on.
 
\bigskip
\noindent\textbf{In the first gauge}\footnote{The author expects that the second gauge will yield the dual system upon completion of all the reduction procedure, but this speculation has not been confirmed, and can be read as equivalent to a wild guess.}
\ $\omega$ is generic and we have
$$
g= \left(\begin{matrix} \rho\Gamma&\rho\Sigma\\ \Sigma&\Gamma\end{matrix}\right)
\left(\begin{matrix} x\mathbf{I}&\omega\\ \0&x^{-1}\mathbf{I}\end{matrix}\right)
=
\left(\begin{matrix} x\rho\Gamma& \rho\Gamma\omega+x^{-1}\rho\Sigma\\ 
x\Sigma& \Sigma\omega+x^{-1}\Gamma\end{matrix}\right).
$$
Let us make the substitution $\Omega=\Sigma\omega+x^{-1}\Gamma$, so that
\be\label{firstgauge}
g
=
\left(\begin{matrix} x\rho\Gamma& \rho\Sigma^{-1}(\Gamma\Omega-x^{-1}\mathbf{I})\\ 
x\Sigma& \Omega\end{matrix}\right).
\ee
As $\omega$ was generic, so is $\Omega$, at this stage, a generic element in $gl(n,\CC)$.

\subsection{Solving the constraint condition}

\medskip\noindent
From (\ref{firstgauge}) we get
$$
gI_{nn}g^\dagger
=
\left(\begin{matrix}
x^2\rho\Gamma^2\rho^\dagger 
- \rho\Sigma^{-1}(\Gamma\Omega-x^{-1}\mathbf{I})
(\Omega^\dagger\Gamma-x^{-1}\mathbf{I})\Sigma^{-1}\rho^\dagger
&\ \ \ \ 
x^2\rho\Sigma\Gamma - \rho\Sigma^{-1}(\Gamma\Omega-x^{-1}\mathbf{I})\Omega^\dagger
\\
x^2\Sigma\Gamma\rho^\dagger 
- \Omega(\Omega^\dagger\Gamma-x^{-1}\mathbf{I})\Sigma^{-1}\rho^\dagger
&
x^2\Sigma^2 - \Omega\Omega^\dagger
\end{matrix}\right),
$$
and imposing the constraint by comparing this with (\ref{ggdag}), we get
\be\label{transomtok}
\left\{{
\begin{array}{lc}
&\Omega\Omega^\dagger= y^2\mathbf{I} + x^2\Sigma^2 = \Lambda^2,\ \hbox{ with }\ 
\Lambda= diag\bigl([y^2+x^2\sinh^2\Delta_i]^{1/2}\bigr),\\
&\\
&\nu=\rho\Sigma^{-1}(y^2\Gamma-x^{-1}\Omega^\dagger),\\
&\\
&x^2\rho\Gamma^2\rho^\dagger - \rho\Sigma^{-1}(\Gamma\Omega
-
x^{-1}\mathbf{I})(\Omega^\dagger\Gamma-x^{-1}\mathbf{I})\Sigma^{-1}\rho^\dagger 
+ y^{-2}\nu\nu^\dagger 
= 
y^{-2}\sigma\sigma^\dagger,
\end{array}}\right.
\ee
from which we deduce that $\Lambda^{-1}\Omega\in U(n)$, or in other words
\be\label{defT}
\Omega=\Lambda T,\quad \hbox{ with }\ T\in U(n),
\ee
and thence
$$
\rho\Sigma^{-1}T^\dagger\Sigma^2T\Sigma^{-1}\rho^\dagger
=
\sigma\sigma^\dagger.
$$

\noindent or
\be\label{finalform}
T^\dagger\Sigma^2T=\Sigma\rho^\dagger\sigma\sigma^\dagger\rho\Sigma.
\ee
(\ref{finalform}) is viewed as being a constraint condition and will be subject to a detailed analysis in Section \ref{KKS}. For the time being $\sigma$ remains as just some constant matrix in $B_n$.

\bigskip\noindent
Having solved for $g$ in (\ref{ggdag}), it follows that $k_R\in K$, defined by $k_R=b_L^{-1}g$, will be in $SU(n,n)$; that is, it satisfies $k_R^\dag I_{n,n}k_R=I_{n,n}$.

We have then, $b_Lk_R=g=k_Lb_R$, with
\be\label{reducedspace}
\ba
&b_L=\left(\begin{matrix} y^{-1}\sigma&y^{-1}\nu\\ \0&y\mathbf{I} \end{matrix}\right)\qquad
\hbox{and}\qquad
\nu=\rho\Sigma^{-1}(y^2\Gamma-x^{-1}\Omega^\dagger),\\
&\\
&b_R=\left(\begin{matrix} x\mathbf{I}&\omega\\ \0&x^{-1}\mathbf{I}\end{matrix}\right)\qquad
\hbox{and}\qquad
 \omega=\Sigma^{-1}(\Omega-x^{-1}\Gamma),\\
&\\
&k_L= \left(\begin{matrix} \rho\Gamma&\rho\Sigma\\ \Sigma&\Gamma\end{matrix}\right),\qquad
\\
&\\
&k_R=b_L^{-1}g=
y^{-1}
\begin{pmatrix}\sigma^{-1}&0\\0&{\mathbf I}\end{pmatrix}
\begin{pmatrix}\rho\Sigma^{-1}&0\\0&{\mathbf I}\end{pmatrix}
\begin{pmatrix}T^\dagger&0\\0&{\mathbf I}\end{pmatrix}
\begin{pmatrix}\Lambda&x\Sigma^2\\ x{\mathbf I}&\Lambda\end{pmatrix}
\begin{pmatrix}\Sigma&0\\0&T\end{pmatrix}.
\ea\ee
We may now use these together with (\ref{heisensymp}) to compute the symplectic structure on the reduced space. We have
\[
db_Lb_L^{-1}=\begin{pmatrix}0&y^{-2}d\nu\\ 0&0\end{pmatrix} 
=\begin{pmatrix}0&y^{-2}d\bigl[\rho\Sigma^{-1}(y^2\Gamma-x^{-1}\Omega^\dagger)\bigr]\\ 0&0\end{pmatrix},
\]
\[
b_R^{-1}db_R = \begin{pmatrix}0 &x^{-1}d\omega\\ 0&0\end{pmatrix}
=\begin{pmatrix}0&x^{-1}d\bigl[\Sigma^{-1}(\Omega-x^{-1}\Gamma)\bigr]\\ 0&0\end{pmatrix},
\]
\[
dk_Lk_L^{-1}=\begin{pmatrix}d\rho\rho^\dagger&\rho d\Delta\\ d\Delta\rho^\dagger&0\end{pmatrix}
\]
\[
\ba
k_R^{-1}dk_R&=
\begin{pmatrix}\Sigma^{-1}d\Sigma&0\\ 0&T^\dagger dT\end{pmatrix}
+
y^{-2}\begin{pmatrix}\Lambda d\Lambda & 2x\Lambda d\Sigma T\\
-xT^\dagger d\Lambda\Sigma & -2x^2T^\dagger\Sigma d\Sigma T + T^\dagger\Lambda d\Lambda T\end{pmatrix}\\
&\quad+
y^{-2}\begin{pmatrix} -\Sigma^{-1}\Lambda dTT^\dagger\Lambda\Sigma & -x\Sigma^{-1}\Lambda dTT^\dagger\Sigma^2T\\
xT^\dagger dTT^\dagger\Lambda\Sigma & x^2T^\dagger dTT^\dagger\Sigma^2T\end{pmatrix}\\
&\quad\quad+
y^{-2}\begin{pmatrix}
\Sigma^{-1}\Lambda T\Sigma\rho^\dagger d(\rho\Sigma^{-1})T^\dagger\Lambda\Sigma
&
x\Sigma^{-1}\Lambda T\Sigma\rho^\dagger d(\rho\Sigma^{-1})T^\dagger\Sigma^2T\\
-x\Sigma\rho^\dagger d(\rho\Sigma^{-1})T^\dagger\Lambda\Sigma
&
-x^2\Sigma\rho^\dagger d(\rho\Sigma^{-1})T^\dagger\Sigma^2T\end{pmatrix}.
\ea
\]
Thus,
\[
\ba
\langle db_Lb_L^{-1}\ \overset{\wedge},\ dk_Lk_L^{-1}\rangle
&=
y^{-2}\langle d\nu\ \overset{\wedge},\ d\Delta\rho^\dagger\rangle\\
&=
\langle\rho^\dagger d\rho\ \overset{\wedge},\ \Sigma^{-1}\Gamma d\Delta - x^{-1}y^{-2}\Sigma^{-1}T^\dagger\Lambda d\Delta\rangle
+
x^{-1}y^{-2}\langle dT\ \overset{\wedge},\ \Lambda\Sigma^{-1}d\Delta\rangle,
\ea
\]
where here several simplifications have been made: \\
(i) as the pairing between $\mathfrak b$ and $\mathfrak k$ is the imaginary part of the trace, we have $\langle A,B\rangle = -\langle A^\dagger , B^\dagger\rangle$, \\
(ii) because of the antisymmetry, denoted by wedge `$\wedge$', all terms involving adjacent diagonal differentials $d\Delta$ are zero. \\
Then,
\[
\ba
&\langle b_R^{-1}db_R\ \overset{\wedge},\ k_R^{-1}dk_R\rangle
=y^{-2}\langle d\omega\ \overset{\wedge},\ -T^\dagger\Sigma d\Lambda + T^\dagger dTT^\dagger\Lambda\Sigma - \Sigma\rho^\dagger d\rho\Sigma^{-1}T^\dagger\Lambda\Sigma + \Sigma^{-1}d\Sigma T^\dagger\Lambda\Sigma\rangle\\
&\qquad=
y^{-2}\langle\Sigma\Lambda d(\Sigma^{-1}\Lambda)\ \overset{\wedge},\ 
dTT^\dagger - T\Sigma\rho^\dagger d\rho\Sigma^{-1}T^\dagger + T\Sigma^{-1}d\Sigma T^\dagger\rangle\\
&\qquad\qquad+
x^{-1}y^{-2}\langle\Lambda\Sigma^{-1}d\Delta\ \overset{\wedge},\ 
T^\dagger dTT^\dagger + - \Sigma\rho^\dagger d\rho\Sigma^{-1}T^\dagger\rangle\\
&\qquad\qquad\quad+
y^{-2}\langle dTT^\dagger\ \overset{\wedge},\ -\Lambda d\Lambda + dTT^\dagger\Lambda^2 - T\Sigma\rho^\dagger d\rho\Sigma^{-1}T^\dagger\Lambda^2 + T\Sigma^{-1}d\Sigma T^\dagger\Lambda^2\rangle\\
&\qquad=
\langle\rho^\dagger d\rho\ \overset{\wedge},\ x^{-1}y^{-2}\Sigma^{-1}T^\dagger\Lambda d\Delta + x^2y^{-2}\Sigma^{-1}T^\dagger\Sigma^2dT\Sigma + \Sigma^{-1}T^\dagger dT\Sigma - \Sigma^{-1}T^\dagger\Sigma^{-1}d\Sigma T\Sigma\rangle\\
&+
\langle dTT^\dagger\ \overset{\wedge},\ -x^2y^{-2}\Sigma d\Sigma + \Sigma^{-1}d\Sigma + x^2y^{-2}dTT^\dagger\Sigma^2 + T\Sigma^{-1}d\Sigma T^\dagger + x^2y^{-2}T\Sigma^{-1}d\Sigma T^\dagger\Sigma^2\rangle\\
&\qquad\qquad+\langle\Sigma^{-1}d\Sigma\ \overset{\wedge},\ T^\dagger\Sigma^{-1}d\Sigma T\rangle + x^{-1}y^{-2}\langle\Lambda\Sigma^{-1}d\Delta\ \overset{\wedge},\ 
 dT\rangle.
\ea
\]
Putting these together, we have
\[
\ba
{[Symp]} &= \langle\rho^\dagger d\rho\ \overset{\wedge},\ \Sigma^{-1}d\Sigma+ \Sigma^{-1}T^\dagger dT\Sigma - \Sigma^{-1}T^\dagger\Sigma^{-1}d\Sigma T\Sigma\rangle
+\langle T^\dagger dT + dTT^\dagger\ \overset{\wedge},\ \Sigma^{-1}d\Sigma\rangle\\
&+
x^2y^{-2}\langle \rho^\dagger d\rho\ \overset{\wedge},\ (\Sigma^{-1}T^\dagger\Sigma^2T\Sigma^{-1})(\Sigma T^\dagger dT\Sigma)\rangle\\
&+
x^2y^{-2}\langle dTT^\dagger\ \overset{\wedge},\ \Sigma d\Sigma + (dTT^\dagger)(\Sigma^2)\rangle
+ 
x^2y^{-2}\langle T^\dagger dT\ \overset{\wedge},\ (\Sigma^{-1}d\Sigma)(T^\dagger\Sigma^2T)\rangle,
\ea
\]
and brackets have been inserted into the last formula to indicate how the identity 
$\langle A,B\rangle = -\langle A^\dagger , B^\dagger\rangle$ can be put to good use.
Use of the constraint (\ref{finalform}) allows the replacement, 
\[
[\Sigma^{-1}T^\dagger\Sigma^2 T\Sigma^{-1},\rho^\dagger d\rho]=[\rho^\dagger\sigma^\dagger\sigma\rho,\rho^\dagger d\rho]=  d(\rho^\dagger\sigma^\dagger\sigma\rho)=d(\Sigma^{-1}T^\dagger\Sigma^2 T\Sigma^{-1}),
\]
 and results in the mutual cancellation of all the terms multiplied by $x^2y^{-2}$. Finally everything may be collected together into the following condensed form
\be\label{nicesymp}
[Symp] = \langle\,\rho^\dagger d\rho\ \overset{\wedge},\ \Sigma^{-1}T^\dagger\Sigma\, d(\Sigma^{-1}T\Sigma)\,\rangle + \langle\, T^\dagger dT + dTT^\dagger\, \overset{\wedge},\ \Sigma^{-1}d\Sigma\,\rangle.
\ee

It is important to ensure that the above structure is invariant  with respect to replacements of $\rho$ by $R\rho$ for which $R$ leaves $\sigma$ invariant; that is $R^\dagger\sigma\sigma^\dagger R=\sigma\sigma^\dagger$.
On the one hand this invariance is guaranteed by the theory of reduction, so that there is nothing to check, but on the other hand it is a useful check that calculations so far have been carried out correctly.
Let us denote by $[Symp]_1$ the problematic term in $[Symp]$. Thus
$$
[Symp]_1 = \langle\rho^\dagger d\rho\ \overset{\wedge}, \ (\Sigma^{-1}T^\dagger\Sigma)d(\Sigma^{-1}T\Sigma)\rangle.
$$ 
The lemma which follows plays a crucial role in the efficient computation of $[Symp]$, not only because it confirms that (\ref{nicesymp}) is well-defined on the reduced space, but also as it hugely simplifies the task, by justifying the making of a simple and convenient choice for $\rho$.

\begin{lemma}\label{rhoinv}
$[Symp]_1$ is invariant with respect to replacement of $\rho$ by $R\rho$ whenever $R$ leaves $\sigma$ invariant.
\end{lemma}
\begin{proof}
Let $\tau:= \Sigma^{-1}T^\dagger\Sigma$. Then
\[
\tau\tau^\dagger = \Sigma^{-1}T^\dagger\Sigma^2T\Sigma^{-1}=\rho^\dagger\sigma\sigma^\dagger\rho
\]
which implies
\be\label{tau}
\tau = \rho^\dagger\sigma Q\qquad\hbox{for some}\ \ Q\in U(n).
\ee
Suppose that $R$ leaves $\sigma$ invariant; that is
\be\label{R}
\left\{{\ba
&R\in U(n)\ \hbox{and}\  R^\dagger\sigma\sigma^\dagger R=\sigma\sigma^\dagger\\
&\Rightarrow\exists S\in U(n)\ \hbox{s.t.}\  R^\dagger\sigma=\sigma S\\
&\Leftrightarrow \qquad \sigma^{-1}R=S^\dagger\sigma^{-1}\\
&\Leftrightarrow S\sigma^{-1}=\sigma^{-1}R^\dagger\ \ \Leftrightarrow R\sigma=\sigma S^\dagger.
\ea}\right.
\ee
Now, rewriting $[Symp]_1$ in terms of $\rho$ and $Q$ instead of in terms of $\rho$ and $\tau$, we have
\[
[Symp]_1=\langle dQQ^\dagger\ \overset{\wedge},\ \sigma^{-1}d\rho\rho^\dagger\sigma\rangle.
\]
Suppose that we replace $\rho$ by $R\rho$, with $R\in U(n)$ s.t. $R^\dagger\sigma\sigma^\dagger R=\sigma\sigma^\dagger$. We have
\[
\tau\tau^\dagger = \rho^\dagger\sigma\sigma^\dagger\rho = \rho^\dagger R^\dagger\sigma\sigma^\dagger R\rho
\]
but now combining (\ref{tau}) and (\ref{R}),
\[
\rho^\dagger R^\dagger\sigma = \rho^\dagger\sigma S = \tau Q^\dagger S,
\quad\hbox{ or }\quad
\tau = \rho^\dagger R^\dagger\sigma S^\dagger Q.
\]
That is, replacing $\rho$ by $R\rho$ is accompanied by the replacement of $Q$ by $S^\dagger Q$. Computing now $[Symp]_1$ at the shifted point $(\tilde\rho,\tilde Q)=(R\rho,S^\dagger Q)$ we have
\[
\ba
{[Symp]}_1(\tilde\rho,\tilde Q)&= \langle d(S^\dagger Q)Q^\dagger S\ \overset{\wedge},\ \sigma^{-1}d(R\rho)\rho^\dagger R^\dagger\sigma\rangle\\
&=
\langle\sigma S^\dagger dQQ^\dagger S\sigma^{-1} - \sigma S^\dagger d(S\sigma^{-1})\ \overset{\wedge},\ dRR^\dagger + Rd\rho\rho^\dagger R^\dagger\rangle\\
&=
\langle R\sigma dQQ^\dagger\sigma^{-1}R^\dagger + dRR^\dagger \ \overset{\wedge},\ dRR^\dagger + Rd\rho\rho^\dagger R^\dagger\rangle\\
&=
\langle dQQ^\dagger \ \overset{\wedge},\ \sigma^{-1}R^\dagger d(R\sigma)  + \sigma^{-1}d\rho\rho^\dagger\sigma\rangle,\quad\hbox{using }\ u(n)^\perp = u(n),\\
&=
\langle dQQ^\dagger \ \overset{\wedge},\ -dSS^\dagger  + \sigma^{-1}d\rho\rho^\dagger\sigma\rangle\\
&=
\langle dQQ^\dagger \ \overset{\wedge},\ \sigma^{-1}d\rho\rho^\dagger\sigma\rangle,
\quad\qquad\qquad\hbox{again using }\ u(n)^\perp = u(n),\\
&=
{[Symp]}_1(\rho,Q)
\ \ \qquad\qquad\qquad\qquad\qquad\qquad\qquad\hbox{as required.}
\ea
\]
\end{proof}

\noindent\textbf{Note:} Making use of the freedom to multiply $\rho$ on the left by any $R$, for which $R^\dagger\sigma\sigma^\dagger R=\sigma\sigma^\dagger$, is the factoring by the remaining gauge freedom which was referred to in the remark between (\ref{ggdag}) and (\ref{firstgauge}).

\section{Imposing the KKS Constraint Condition}\label{KKS}

The KKS constraint condition is conveniently expressed, not as an explicit one on $\sigma$, but rather, \emph{equivalently}, as one on the product $\sigma\sigma^\dagger$. Thus, fixing $\alpha\in\RR\backslash\{0,\pm1\}$ and $\hat v\in\CC^n$, $\sigma$ is required to satisfy
\be
\label{KKScondition}
\sigma\sigma^\dagger = \alpha^2{\mathbf I} + \epsilon\hat v\hat v^\dagger,
\ee
where $\epsilon = -\sgn(\log(\alpha^2))$, i.e. $\epsilon$ is $+1$ if $\alpha^2<1$ and $-1$ if $\alpha^2>1$.
If we allow $\alpha$ to be arbitrary, $\hat v$ is restricted by the condition that $\det\sigma=1$, but direct use of this condition is not needed and it becomes taken care of automatically, so no more will be said about it. Moreover, the result we are aiming for is invariant with respect to $\alpha\mapsto\alpha^{-1}$, so we may safely assume that $\alpha^2<1$ and take $\epsilon=1$.

Together with (\ref{KKScondition}), we must solve the condition given by (\ref{finalform}). It is convenient to introduce the vectors $\tilde v:=\rho^\dagger\hat v$ and $v:=\Sigma\tilde v$. Applying (\ref{KKScondition}),
\be\label{withv}
\Sigma\rho^\dagger\sigma\sigma^\dagger\rho\Sigma=\alpha^2\Sigma^2+vv^\dagger
\ee
and it may be argued, using standard gauge-freedom arguments, that 
$v_i\in\RR_{\geq0}$ : concretely, recalling that $\rho$ is only determined modulo right-multiplication by $\TT$, we may modify $\rho$ by multiplying it on the right by a suitable diagonal matrix so that each of the components $v_i$ is real and non-negative.  Subtracting $\lambda$ from both sides of (\ref{finalform}), then using (\ref{withv}) and taking determinants, we have
\[
\det(\Sigma^2-\lambda)=\det(\alpha^2\Sigma^2-\lambda + vv^T)
=
\det(\alpha^2 \Sigma^2-\lambda)[1+v^T(\alpha^2\Sigma^2-\lambda)^{-1}v],
\]
thus,
\be\label{crux}
1+v^T(\alpha^2\Sigma^2-\lambda)^{-1}v = 
\frac{\det(\Sigma^2-\lambda)}
{\det(\alpha^2\Sigma^2-\lambda)}\qquad\forall\lambda\,.
\ee
The residue at $\lambda=\alpha^2\Sigma_k^2$ gives
\be
v_k^2 = \left({\prod_{j\neq k}(\alpha^2\Sigma_j^2-\alpha^2\Sigma_k^2)}\right)^{-1}\prod_{i=1}^n(\Sigma_i^2-\alpha^2\Sigma_k^2)\ .
\ee
Next let's solve for $T$ in (\ref{finalform}), so
$T$ must satisfy
\be\label{nowT}
T^\dagger\Sigma^2T=\alpha^2\Sigma^2 + vv^T.
\ee
This determines $T$ only up to left multiplication by $\TT$. 
Let $T^\dagger=(\t_1,\dots,\t_n)$, with $\mathbf{t}_k\in\CC^n$ satisfying $\t_i^\dagger\t_j=\delta_{ij}$. We get $T^\dagger\Sigma^2T=\sum \Sigma_i^2\t_i\t_i^\dagger$ so that, from (\ref{nowT}), we have
$$
\sum_{i=1}^n\Sigma_i^2\t_i\t_i^\dagger = \alpha^2\Sigma^2 + vv^T.
$$
Letting both sides of the above equation act on the vector $\t_i$, and using the orthonormal property of the vectors $\{\t_k\}$ we obtain
$$
\Sigma_i^2\t_i = \alpha^2\Sigma^2\t_i + (v^T\t_i)v,
$$
i.e.
$$
(\Sigma_i^2\mathbf{I}-\alpha^2 \Sigma^2)\t_i
=
(v^T\t_i)v
\quad\Rightarrow\quad
\t_i
=
(v^T\t_i)(\Sigma_i^2\mathbf{I}-\alpha^2 \Sigma^2)^{-1}v.
$$
In fact, writing $\hat\t_i=(\Sigma_i^2\mathbf{I}-\alpha^2 \Sigma^2)^{-1}v$, $\t_i$ is $\hat\t_i$ normalised. We may check that $\{\hat\t_1,\hat\t_2,\dots,\hat\t_n\}$ is an orthogonal set. Hence the vectors $\t_i$ are completely defined, up to multiplication by an element of $U(1)$
\be\label{explicitT}
\ba
\t_i&=\theta_i\tilt_i,\quad \theta_i\in U(1)\\
\hbox{with}\qquad
\tilt_i&=(\hat\t_i^\dagger\hat\t_i)^{-\frac{1}{2}}\hat\t_i \\
&= 
\bigl[v^T(\Sigma_i^2\mathbf{I}-\alpha^2\Sigma^2)^{-2}v\bigr]^{-\frac{1}{2}}(\Sigma_i^2-\alpha^2\Sigma^2)^{-1}v.
\ea
\ee
Hence we have 
\be
T = P \tilde T ,
\ee
with $\tilde T^\dagger=(\tilt_1,\dots,\tilt_n)$,  
with the vectors $\tilt_i$ defined in (\ref{explicitT}), and 
\be
\TT\owns P=\exp(ip)\, ,\qquad\hbox{with}\qquad p=diag(p_1,\dots,p_n).
\ee 
Let us note that $\tilde T$ is \emph{real}; that is $\tilde T\in O(n)$.

Due to Lemma \ref{rhoinv}, in computing the explicit form of $[Symp]$, we may use the most convenient representative of $\rho$ for our needs, and indeed $\rho$ itself does not appear, but rather only the combination $\rho^\dagger d\rho$. 
How, though, is $\rho$ defined? Well, it must satisfy the condition 
$\rho^\dagger\sigma\sigma^\dagger\rho=\alpha^2{\mathbf I} + \tv\tv^\dagger$, 
with $\tv=\Sigma^{-1}v\in \RR^n$. From this it follows that
\[
d(\tv\tv^\dagger)=[\rho^\dagger\sigma\sigma^\dagger\rho,\rho^\dagger d\rho] = [\tv\tv^\dagger,\rho^\dagger d\rho].
\]
Hence we may use for $\rho^\dagger d\rho$, just any skew symmetric matrix which satisfies this condition, and for this there exists a very simple choice. First of all let us notice that $\tv^\dagger\tv$ is constant: we have 
\[
\tilde v =\rho^\dagger\hat v\ \Rightarrow |\tilde v|^2 = |\hat v|^2.
\]
Hence $d(\tv^\dagger\tv)=0$. Consider now
\[
\ba
{[}\tv\tv^\dagger, \tv d\tv^\dagger - d\tv\tv^\dagger{]} &= (\tv^\dagger\tv)\tv d\tv^\dagger - (\tv^\dagger d\tv)\tv\tv^\dagger - (d\tv^\dagger\tv)\tv\tv^\dagger + (\tv^\dagger\tv)d\tv\tv^T\\ 
&= |\tv|^2( \tv d\tv^\dagger + d\tv\tv^\dagger)
=
 |\tv|^2 d(\tv\tv^\dagger).
\ea
\]
A convenient choice then for $\rho^\dagger d\rho$ is
\be\label{rhochoice}
\rho^\dagger d\rho = |\hat v|^{-2}(\tv d\tv^\dagger - d\tv\tv^\dagger) =  |\hat v|^{-2}(\tv d\tv^T - d\tv\tv^T) .
\ee
With this choice, $\rho^\dagger d\rho$ is \emph{real}.\\

We now have everything we need to compute $[Symp]$ explicitly:

\[
T^\dagger dT = \tilde T^\dagt d\tilde T + i\tilde T^\dagt dp\tilde T\qquad \hbox{and}\qquad
dTT^\dagger = idp+ Pd\tilde T\tilde T^\dagt P^\dagger,
\]
from which we obtain
\[
[Symp]_2= \langle T^\dagger dT + dTT^\dagger \ \overset{\wedge},\ \Sigma^{-1}d\Sigma\rangle
=
tr\bigl(dp\wedge ( \Sigma^{-1}d\Sigma + \tilde T\Sigma^{-1}d\Sigma\tilde T^\dagt) \bigr).
\]
We have
\[
\ba
\Sigma^{-1}T^\dagger\Sigma d(\Sigma^{-1}T\Sigma) &=
\Sigma^{-1}d\Sigma + \Sigma^{-1}T^\dagger dT\Sigma - \Sigma^{-1}T^\dagger\Sigma^{-1}d\Sigma T\Sigma\\
&=
\Sigma^{-1}d\Sigma + \Sigma^{-1}\tilde T^\dagt d\tilde T\Sigma + i \Sigma^{-1}\tilde T^\dagt dp\tilde T\Sigma - \Sigma^{-1}\tilde T^\dagt\Sigma^{-1}d\Sigma \tilde T\Sigma\ ,
\ea
\]
from which we obtain, substituting from (\ref{rhochoice}) and using the fact that with this choice $\rho^\dagger d\rho$ is real,
\[
\ba
{[Symp]}_1 &= |\hat v|^{-2}\langle \tv d\tv^T - d\tv\tv^T\ \overset{\wedge},\  i \Sigma^{-1}\tilde T^\dagt dp\tilde T\Sigma\rangle\\
&=
|\hat v|^{-2}tr\bigl(\tilde T\Sigma(\tv d\tv^T- d\tv\tv^T)\Sigma^{-1}\tilde T^\dagt\wedge dp\bigr).
\ea
\]
In order to simplify the formula for $[Symp]_1$ we need the diagonal part of the matrix \(\tilde T\Sigma(\tilde vd\tilde v^T - d\tilde v\tilde v^T)\Sigma^{-1}\tilde T^T\), so let's compute
\[
\ba
\Bigl[\tilde T\Sigma(\tilde vd\tilde v^T - d\tilde v\tilde v^T)\Sigma^{-1}\tilde T^T\Bigr]_{ii}
&=\\
[v^T(\Sigma_i^2-\alpha^2\Sigma^2)^{-2}v]^{-1}
\sum_{k=1}^n\sum_{l=1}^n&\,(\Sigma_i^2-\alpha^2\Sigma_k^2)^{-1}(\Sigma_i^2-\alpha^2 \Sigma_l^2)^{-1}\, \Sigma_k^2\, \bigl[
\tv_k^2\tv_ld\tv_l - \tv_l^2\tv_kd\tv_k\bigr].
\ea
\]
It is convenient to replace $\Sigma_i^2$ by $\lambda$ in this expression and 
$\tv_kd\tv_k$ by  $\half d(\tilde v_k^2)$, so that we have
\[
\half[v^T(\lambda-\alpha^2\Sigma^2)^{-2}v]^{-1}
\displaystyle{\sum_{k=1}^n\sum_{l=1}^n}\,(\lambda-\alpha^2\Sigma_k^2)^{-1}(\lambda-\alpha^2\Sigma_l^2)^{-1}
\bigl[-\Sigma_k^2d(\tilde v_k^2)\tilde v_l^2
+
\Sigma_k^2\tilde v_k^2d(\tilde v_l^2)\bigr],
\]
which we may rewrite as 
\[
\half[v^T(\lambda-\alpha^2\Sigma^2)^{-2}v]^{-1}
\displaystyle{\sum_{k=1}^n\sum_{l=1}^n}\,(\lambda-\alpha^2\Sigma_k^2)^{-1}(\lambda-\alpha^2\Sigma_l^2)^{-1}
d(\tilde v_k^2)(\tilde v_l^2)\bigl[\Sigma_l^2-\Sigma_k^2\bigr].
\]
Now 
\[
\frac{1}{\lambda-\alpha^2\Sigma_k^2} - \frac{1}{\lambda-\alpha^2\Sigma_l^2}
=
\frac{\alpha^2(\Sigma_k^2-\Sigma_l^2)}{(\lambda-\alpha^2\Sigma_k^2)(\lambda-\alpha^2\Sigma_l^2)}
\]
so we obtain
\[
\ba
&
\frac{1}{2\alpha^2}\frac{1}{[v^T(\lambda-\alpha^2\Sigma^2)^{-2}v]}
\left(
\displaystyle{\sum_{k=1}^nd(\tilde v_k^2)\sum_{l=1}^n
(\lambda-\alpha^2\Sigma_l^2)^{-1}
\tilde v_l^2
-
\displaystyle{\sum_{k=1}^n}
(\lambda-\alpha^2\Sigma_k^2)^{-1}
d(\tilde v_k^2)\sum_{l=1}^n}\tilde v_l^2\right)\\
&=
\frac{1}{2\alpha^2}\frac{1}{[v^T(\lambda-\alpha^2\Sigma^2)^{-2}v]}
\left(
\displaystyle{d|\tilde v|^2\sum_{l=1}^n\,\frac{\tilde v_l^2}{
\lambda-\alpha^2\Sigma_l^2}
-
|\tilde v|^2\sum_{k=1}^n
\frac{d(\Sigma_k^{-2}v_k^2)}{\lambda-\alpha^2\Sigma_k^2}}\right)\\
&=
-
\frac{1}{2\alpha^2}|\hat v|^2 \frac{1}{[v^T(\lambda-\alpha^2\Sigma^2)^{-2}v]}
\displaystyle{\sum_{k=1}^n\,\frac{d(\tilde v_k^2)}{\lambda-\alpha^2\Sigma_k^2}}\ ,
\qquad
\hbox{as}\quad |\tilde v|=|\hat v| \Rightarrow d|\tilde v|^2=0.
\ea
\]
Thus
\[
[Symp]_1=\left.\frac{1}{2\alpha^2}\displaystyle{\sum_{i=1}^n}\left(\displaystyle{[v^T(\lambda-\alpha^2\Sigma^2)^{-2}v]^{-1}\sum_{k=1}^n(\lambda-\alpha^2\Sigma_k^2)^{-1}dp_i\wedge d(\Sigma_k^{-2}v_k^2)}\right)\right|_{\textstyle{\lambda=\Sigma_i^2}}
\]

The $\tilde T\Sigma^{-1}d\Sigma\tilde T^T$ term in $[Symp]_2$ shakes down to the following expression
\[
\bigl(\tilde T\Sigma^{-1}d\Sigma\tilde T^T\bigr)_{ii}=
\frac{1}{2\alpha^2}
\left.\left([v^T(\lambda-\alpha^2\Sigma^2)^{-2}v]^{-1}\displaystyle{\sum_{k=1}^n\Sigma_k^{-2}v_k^2d[(\lambda-\alpha^2\Sigma_k^2)^{-1}]}\right)\right|_{\textstyle{\lambda=\Sigma_i^2}},
\]
and hence
\[
[Symp]_2=\sum_{i=1}^ndp_i\wedge\Sigma_i^{-1}d\Sigma_i
+
\left.\frac{1}{2\alpha^2}
\displaystyle{ \sum_{i=1}^n}\left(\displaystyle{[v^T(\lambda-\alpha^2\Sigma^2)^{-2}v]^{-1} \sum_{k=1}^n \Sigma_k^{-2} v_k^2 dp_i\wedge d[(\lambda-\alpha^2\Sigma_k^2)^{-1}]}\right)\right|_{\textstyle{\lambda=\Sigma_i^2}}.
\]
Putting these together we have
\[
\ba
{[Symp]}
&=
\sum_{i=1}^n dp_i\wedge\Sigma_i^{-1}d\Sigma_i
+
\left.\frac{1}{2\alpha^2}
\displaystyle{ \sum_{i=1}^n}\left(\displaystyle{[v^T(\lambda-\alpha^2\Sigma^2)^{-2}v]^{-1} \sum_{k=1}^n dp_i\wedge d\left[\frac{\tilde v_k^2}{\lambda-\alpha^2\Sigma_k^2}\right]}\right)\right|_{\textstyle{\lambda=\Sigma_i^2}}\\
&=
\sum_{i=1}^n\ dp_i\wedge\Sigma_i^{-1}d\Sigma_i\\
&\qquad\qquad
+
\frac{1}{2\alpha^2}
\displaystyle{\sum_{i=1}^n}\ dp_i\wedge 
\left.\left(
[v^T(\lambda-\alpha^2\Sigma^2)^{-2}v]^{-1}
d\Bigl[v^T(\lambda-\alpha^2\Sigma^2)^{-1}\Sigma^{-2}v\Bigr]
\right)\right|_{\textstyle{\lambda=\Sigma_i^2}}.
\ea
\]
This is a very useful expression, as we may find a concise form for the combination\\
$v^T(\lambda-\alpha^2\Sigma^2)^{-1}\Sigma^{-2}v$.

\subsection{Computing a convenient expression for $d\left[v^T(\lambda-\alpha^2\Sigma^2)^{-1}\Sigma^{-2}v\right]$:}

Consider
\[
\ba
|\tilde v|^2 
&= v^T\Sigma^{-2}v=v^T\Sigma^{-2}(\lambda-\alpha^2\Sigma^2)(\lambda-\alpha^2\Sigma^2)^{-1}v\\
&= \lambda v^T\Sigma^{-2}(\lambda-\alpha^2\Sigma^2)^{-1}v 
- \alpha^2 v^T(\lambda-\alpha^2\Sigma^2)^{-1}v\\
\Rightarrow
0=d(|\tilde v|^2) &= \lambda\, d\Bigl(v^T\Sigma^{-2}(\lambda-\alpha^2\Sigma^2)^{-1}v\Bigr)
- \alpha^2\, d\Bigl(v^T(\lambda-\alpha^2\Sigma^2)^{-1}v\Bigr).
\ea
\]
Thus
\be
d\Bigl(v^T\Sigma^{-2}(\lambda-\alpha^2\Sigma^2)^{-1}v\Bigr) 
= \frac{\alpha^2}{\lambda}d\Bigl(v^T(\lambda-\alpha^2\Sigma^2)^{-1}v\Bigr)
\qquad\forall\lambda.
\ee
From (\ref{crux}) we already have
\[
v^T(\lambda-\alpha^2\Sigma^2)^{-1}v = 
1 - \frac{\det(\lambda-\Sigma^2)}{\det(\lambda-\alpha^2\Sigma^2)}\ ,
\]
which implies
\be
d\bigl[v^T(\lambda-\alpha^2\Sigma^2)^{-1}v\bigr]
=
2\frac{\det(\lambda-\Sigma^2)}{\det(\lambda-\alpha^2\Sigma^2)}
\sum_{k=1}^n\Bigl((\lambda-\Sigma_k^2)^{-1} - \alpha^2(\lambda-\alpha^2\Sigma_k^2)^{-1}\Bigr)\Sigma_kd\Sigma_k\ ,
\ee
and also
\be
\ba
v^T(\lambda-\alpha^2\Sigma^2)^{-2}v &= \frac{d}{d\lambda}v^T(\alpha^2\Sigma-\lambda)^{-1}v \\
&=
\frac{\det(\lambda-\Sigma^2)}{\det(\lambda-\alpha^2\Sigma^2)}
\sum_{k=1}^n\Bigl((\lambda-\Sigma_k^2)^{-1} - (\lambda-\alpha^2\Sigma_k^2)^{-1}\Bigr).
\ea
\ee
Putting all these together we obtain
\be
\ba
&\Bigl(v^T(\lambda-\alpha^2\Sigma^2)^{-2}v\Bigr)^{-1}
d\Bigl(v^T\Sigma^{-2}(\lambda-\alpha^2\Sigma^2)^{-1}v\Bigr) 
=\\
&\qquad
2
\frac{\alpha^2}{\lambda}
\left(\sum_{l=1}^n\Bigl[(\lambda-\Sigma_l^2)^{-1} - (\lambda-\alpha^2\Sigma_l^2)^{-1}\Bigr]\right)^{-1}
\left(\sum_{k=1}^n\Bigl[(\lambda-\Sigma_k^2)^{-1} - \alpha^2(\lambda-\alpha^2\Sigma_k^2)^{-1}\Bigr]\Sigma_kd\Sigma_k\right)
\ea
\ee
and taking the limit as $\lambda\rightarrow\Sigma_i^2$ gives
\be
\left.\Bigl(v^T(\lambda-\alpha^2\Sigma^2)^{-2}v\Bigr)^{-1}
d\Bigl(v^T\Sigma^{-2}(\lambda-\alpha^2\Sigma^2)^{-1}v\Bigr)\right|_{\textstyle{\lambda=\Sigma_i^2}}
=
2\alpha^2\Sigma_i^{-1}d\Sigma_i.
\ee
Hence we arrive at the formula
\be\label{finalsymp}
[Symp] = 2\sum_{i=1}^ndp_i\wedge\Sigma_i^{-1}d\Sigma_i,
\ee
and so find ourselves in the happy position of having discovered canonical coordinates on the reduced space via our first natural choice. In standard form, the canonical coordinates $(q,p)$ are of course given by $\Sigma_i=\exp(q_i)$.

\section{Commuting Hamiltonians} 
The functions in the commuting family of Proposition \ref{commies} are all invariant with respect to the actions of $K_+$ and so induce a commuting family on the reduced space. The simplest one of these is $\Phi_1$. We have
\[
\ba
\omega^\dagger\omega &= (\Omega^\dagger-x^{-1}\Gamma)\Sigma^{-2}(\Omega-x^{-1}\Gamma)\\
&=
T^\dagger\Lambda^2\Sigma^{-2}T-x^{-1}(T^\dagger\Lambda\Sigma^{-2}\Gamma+\Gamma\Sigma^{-2}\Lambda T) +x^{-2}\Sigma^{-2}\Gamma^2\\
&=
T^\dagger(y^2\Sigma^{-2}+x^2)T - x^{-1}(T^\dagger\Lambda\Sigma^{-2}\Gamma+\Gamma\Sigma^{-2}\Lambda T) + x^{-2}(\Sigma^{-2}+\mathbf{I}).
\ea
\]
The definition
\be
\ba
\Phi_1&= -\half\tr (g^\dagger I_{nn}gI_{nn}) \\ 
&= {\half}tr(\omega^\dagger\omega) -{\half} n(x^2+x^{-2})
\\
&=
{\half}(y^2+x^{-2})tr(\Sigma^{-2})  - {\half}x^{-1}tr\Lambda\Gamma\Sigma^{-2}\tilde T^T(P+P^*)
\ea
\ee
results in
\be\label{sigmaformofHam}
\ba
\Phi_1 &=
{\half}(x^{-2}+y^2)\sum_{i=1}^n\Sigma_i^{-2}\\ 
&\qquad- x^{-1}\sum_{i=1}^n(\cos p_i)\sqrt{1+\Sigma_i^{-2}}\sqrt{x^2+y^2\Sigma_i^{-2}}\displaystyle{\prod_{k\neq i}}\frac{\sqrt{\alpha^{-1}\Sigma_k^2-\alpha\Sigma_i^2}\sqrt{\alpha\Sigma_k^2-\alpha^{-1}\Sigma_i^2}}{(\Sigma_k^2-\Sigma_i^2)}\ .
\ea
\ee

At this juncture it is easy to see that there is a $BC_n$ Weyl group action on the reduced space leaving the Hamiltonian $\Phi_1$ invariant. This action is generated by the symplectic lifts of the operations
\be
\ba
&\Delta_i\mapsto -\Delta_i\qquad\qquad\quad\  i=1,\dots,n\\
&(\Delta_i,\Delta_k)\mapsto (\Delta_k,\Delta_i)\qquad i,k = 1,\dots, n
\ea
\ee
which, comparing with (\ref{finalsymp}), is represented in terms of the coordinates $\{(\Sigma,p)\}$ by
\be
\ba
&(\Sigma_i,p_i) \mapsto (-\Sigma_i,+p_i)\qquad\qquad\qquad\ i=1,\dots,n\\
&(\Sigma_i,\Sigma_k,p_i,p_k)\mapsto (\Sigma_k,\Sigma_i,p_k,p_i)\qquad i,k = 1,\dots, n
\ea
\ee
The invariance claim is now revealed by direct inspection of (\ref{sigmaformofHam}).

\noindent\textbf{Remark:}
In the context of the procedure used to arrive at this point, going back to where a restriction was made to an open subset of the unreduced phase space, defined by the requirement that $\Delta$ lie within some particular fixed Weyl chamber, the above action of the Weyl group acts by changing the Weyl chamber defined by (\ref{Weylchambercondition}) to some other one, without the need for any change in the reduction argument.

Equation (\ref{sigmaformofHam}) may be rewritten as
\be\label{qformofHam}
\begin{aligned}
&\Phi_1 = {\half}(x^{-2}+y^2)\sum_{i=1}^ne^{-2q_i}\\
&-
\sum_{i=1}^n (\cos p_i)
\left[{ 1 + \left({1+ \frac{y^2}{x^2}}\right)e^{-2q_i} + \frac{y^2}{x^2}e^{-4q_i}}\right]^{\half}
\prod_{k\neq i}\left[ 1 - \frac{(\alpha-\alpha^{-1})^2}{4\sinh^2(q_i-q_k)}\right]^{1/2}
\end{aligned}
\ee
Aside from the first extra term in the Hamiltonian (\ref{qformofHam}), it has the standard form of a Ruijsenaars type model, and it is invariant under the action of the Weyl group for the $BC_n$ root system, whilst it depends on the three independent parameters $\alpha, x, y$.

\bigskip

\subsection{Linearisation}

``Linearisation'' is the passage from the Heisenberg double to the cotangent bundle, which is the same as the semi-direct product $K\ltimes\b = K\ltimes\k^*$. This is equivalent to assuming that the $B$-component of $g$ consists of elements of the form ${\bf I} + tX + O(t^2)$ with $X\in\b$, and $t$ very small.

The cotangent bundle limit involves the substitutions
\[
x=\exp(t\xi),\quad y=\exp(t\eta),\quad \alpha =\exp(t\zeta),\quad p=t\pi
\]
and re-scaling of the symplectic structure. We discover that we must replace $\pi$ by $\Gamma^{-1}\Sigma\hat p$ in order to have the Hamiltonian in the form $\half|\hat p|^2 + V(\hat q)$. This results in the symplectic structure having the form $\sum d\hat p_i\wedge d\hat q_i$, for $\hat q=\Delta$, and the Hamiltonian is the same as the three-parameter hyperbolic $BC_n$ Sutherland Hamiltonian from \cite{Inozemtsev} and \cite{FehPus}.

Explicitly :
\[\ba
&x^{-2}+y^2 = 2 + 2t(\eta-\xi) + 2 t^2(\eta^2-\xi^2) +\cdots\ ,  \\
&\cos(t\pi_i) = 1 - \half t^2\pi_i^2+\cdots\ ,  \\
&\left(1+\frac{y^2}{x^2}\Sigma_i^{-2}\right)^{1/2} = 
(1+\Sigma_i^{-2})^{1/2}\bigl(1+2t(\eta-\xi)\Gamma_i^{-2} + 2t^2(\eta-\xi)^2\Gamma_i^{-2}+\cdots\bigr)^{1/2}\\
&\qquad\qquad\qquad=
(1+\Sigma_i^{-2})^{1/2}\bigl(1+t(\eta-\xi)\Gamma_i^{-2}+t^2(\eta-\xi)^2\Gamma_i^{-2} 
- \half t^2(\eta-\xi)^2\Gamma_i^{-4}+\cdots\bigr)\ ,  \\
&(\Sigma_k^2-\Sigma_i^2)^{-1}\bigl(\alpha^{-1}\Sigma_k - \alpha\Sigma_i^2\bigr)^{1/2}
(\alpha\Sigma_k - \alpha^{-1}\Sigma_i^2\bigr)^{1/2}\\
&\qquad\qquad=
(\Sigma_k^2-\Sigma_i^2)^{-1}\Bigl(\Sigma_k^2-\Sigma_i^2- t\zeta(\Sigma_k^2+\Sigma_i^2) + {\half}t^2\zeta^2(\Sigma_k^2 - \Sigma_i^2)+\cdots\Bigr)^{1/2}\times\\
&\qquad\qquad\qquad\qquad\qquad\qquad\qquad\Bigl(\Sigma_k^2-\Sigma_i^2 + t\zeta(\Sigma_k^2 + \Sigma_k^2) + {\half}t^2\zeta^2(\Sigma_k^2-\Sigma_i^2)+\cdots\Bigr)^{1/2}\\
&\qquad\qquad=
(\Sigma_k^2-\Sigma_i^2)^{-1}\Bigl((\Sigma_k^2-\Sigma_i^2)^2 + t^2\zeta^2(\Sigma_k^2 - \Sigma_i^2)^2 - t^2 \zeta^2(\Sigma_k^2+\Sigma_i^2)^2 +\cdots\Bigr)^{1/2}\\
&\qquad\qquad=
\left(1
- 4t^2\zeta^2\frac{\Sigma_i^2\Sigma_k^2}{(\Sigma_k^2-\Sigma_i^2)^2}+\cdots\right)^{1/2}\\
&\qquad\qquad=
\left(1
- 2t^2\zeta^2\frac{\Sigma_i^2\Sigma_k^2}{(\Sigma_k^2-\Sigma_i^2)^2}+\cdots\right)\ .
\ea
\]
Thus,
\[
\ba
\Phi_1(t)&= \Bigl(1+t(\eta-\xi) + t^2(\eta^2-\xi^2)+\cdots\Bigr)\sum_{i=1}^n\Sigma_i^{-2}\\
&-\sum_{i=1}^n\bigl(1-\half t^2\pi_i^2+\cdots\bigr)(1+\Sigma_i^{-2})\Bigl[1+t(\eta-\xi)\Gamma_i^{-2} + t^2(\eta-\xi)^2\Gamma_i^{-2} - \half t^2(\eta-\xi)^2\Gamma_i^{-4}+\cdots\Bigr]\times  \\
&\qquad\qquad\times \prod_{k\neq i}\left(1- 2t^2\zeta^2\frac{\Sigma_i^2\Sigma_k^2}{(\Sigma_i^2-\Sigma_k^2)^2}\cdots\right)\\
&=
H_0 + tH_1 + t^2H_2+\cdots
\ea
\]
with
\[\ba
H_0 &= -n,\quad
H_1 = 0,\\
H_2 &= {\half}\sum_{i=1}^n\left(\frac{\Gamma_i\pi_i}{\Sigma_i}\right)^2 + 
[(\eta^2-\xi^2) - (\eta-\xi)^2]\sum_{i=1}^n\Sigma_i^{-2}\\ 
&\qquad\qquad+ {\half}(\eta-\xi)^2\sum_{i=1}^n\Sigma_i^{-2}\Gamma_i^{-2}
+
2\zeta^2\sum_{i=1}^n\sum_{k\neq i}\frac{\Gamma_i^2\Sigma_k^2}{(\Sigma_i^2-\Sigma_k^2)^2}.
\ea
\]

Let us rescale the symplectic structure, which is equivalent to rescaling time, thus
\[
[Symp] \mapsto \widehat{[Symp]}=t[Symp],
\]
and
\[
\widehat{[Symp]} = t^2\sum_{i=1}^nd\pi_i\wedge\Sigma_i^{-1}d\Sigma_i
=
t^2\sum_{i=1}^nd\left(\frac{\Gamma_i\pi_i}{\Sigma_i}\right)\wedge\Gamma_i^{-1}d\Sigma_i
=
t^2\sum_{i=1}^nd\hat p_i\wedge d\Delta_i = t^2 \cS, \quad\hbox{say}.
\]
On the other hand, the Hamiltonian vector-field $\XX_\Phi$ is defined, with respect to $\widehat{[Symp]}$ by
\[
t^2dH_2 + O(t^3) = d\Phi = \widehat{[Symp]}(\ \cdot\ ,\XX_\Phi) =  t^2\cS(\ \cdot\ ,\XX_{H_2}) + O(t^3)
\]
Hence, taking the limit $t\to0$, we have
\[
dH_2=\cS(\ \cdot\ ,\XX_{H_2})
\]
and so the limit is $H_2$, with canonical coordinates $\hat p_i=\Sigma_i^{-1}\Gamma_i\pi_i$ and $\hat q_i=\Delta_i$. $H_2$ is recognised as the general hyperbolic $BC_n$ Sutherland Hamiltonian. To see this, the only difficult term is the last one. We have
\[
2\sum_{i=1}^n\sum_{j\neq i}\frac{\Gamma_i^2\Sigma_j^2}{(\Sigma_i^2-\Sigma_j^2)^2}
=
\sum_{i,j,\ i\neq j}\frac{\Gamma_i^2\Sigma_j^2 + \Sigma_i^2\Gamma_j^2}{(\Sigma_i^2-\Sigma_j^2)^2}.
\]
\[
\Sigma_i^2-\Sigma_j^2 = \Sigma_i^2\Gamma_j^2 - \Sigma_j^2\Gamma_i^2
=
(\Sigma_i\Gamma_j+\Sigma_j\Gamma_i)(\Sigma_i\Gamma_j-\Sigma_j\Gamma_i)
=
\sinh(\hat q_i+\hat q_j)\sinh(\hat q_i-\hat q_j),
\]
and
\[
\Sigma_i^2\Gamma_j^2 + \Sigma_j^2\Gamma_i^2 = \half\Bigl(
(\Sigma_i\Gamma_j+\Sigma_j\Gamma_i)^2 + (\Sigma_i\Gamma_j-\Sigma_j\Gamma_i)^2\Bigr)
=
\half\Bigl(\sinh^2(\hat q_i+\hat q_j) + \sinh^2(\hat q_i-\hat q_j)\Bigr).
\]
Hence the last term is
\[
\zeta^2\sum_{i,j\ i\neq j}\left[\frac{1}{\sinh^2(\hat q_i+\hat q_j)} + \frac{1}{\sinh^2(\hat q_i-\hat q_j)}\right].
\]
The other terms are easily expressed as functions of $\hat q$ and we have
\[
\ba
H_2&= {\half}\sum_{i=1}^n\hat p_i^2 + c_1\sum_{i=1}^n\frac{1}{\sinh^2\hat q_i} +
c_2\sum_{i=1}^n\frac{1}{\sinh^2(2\hat q_i)}\\ 
&\qquad\qquad\qquad + c_3\sum_{i,j\ i\neq j}\left[\frac{1}{\sinh^2(\hat q_i+\hat q_j)} + \frac{1}{\sinh^2(\hat q_i-\hat q_j)}\right].
\ea
\]

\end{document}